\documentclass[conference,10pt]{IEEEtran}

\usepackage[utf8]{inputenc}
\usepackage[T1]{fontenc}

\ifCLASSINFOpdf

\else
\fi

\usepackage{graphicx}
\usepackage{algorithm}
\usepackage{epstopdf}
\usepackage{footnote}
\usepackage{tikz}
\usepackage{subfigure}
\usepackage{cuted}

\graphicspath{{./Figures/}}

\usepackage[cmex10]{amsmath}
\usepackage{amsfonts,latexsym,amssymb,amsthm}
\usepackage{epsfig,graphics,color}
\usepackage{setspace}
\usepackage{multirow}
\usepackage{cite}
\usepackage{algorithm, algpseudocode}
\usepackage{mathrsfs}
\usepackage{bm}
\usepackage{enumitem}
\usepackage[normalem]{ulem}

\setlist[itemize]{noitemsep, topsep=0pt}

\newcommand{\N}{\mathcal{N}}

\newtheorem{theorem}{Theorem}
\newtheorem{lemma}{Lemma}

\newtheorem{remark}{Remark}

\usepackage{soul}

\title{Efficient Hamiltonian Reduction for Quantum Annealing on SatCom Beam Placement Problem}
\author{Thinh Q. Dinh$^\dagger$, Son Hoang Dau$^\ddagger$, Eva Lagunas$^\dagger$,  
and Symeon Chatzinotas$^\dagger$\\
$^\dagger$University of Luxembourg, Luxembourg\\
$^\ddagger$RMIT University, Melbourne}

\begin{document}

\maketitle
\begin{abstract}
Beam Placement (BP) is a well-known problem in Low-Earth Orbit (LEO) satellite communication (SatCom) systems, which can be modelled as an NP-hard clique cover problem. Recently, quantum computing has emerged as a novel technology which revolutionizes how to solve challenging optimization problems by formulating Quadratic Unconstrained Binary Optimization (QUBO), then preparing Hamiltonians as inputs for quantum computers. In this paper, we study how to use quantum computing to solve BP problems. However, due to limited hardware resources, existing quantum computers are unable to tackle large optimization spaces. Therefore, we propose an efficient Hamiltonian Reduction method that allows quantum processors to solve large BP instances encountered in LEO systems. We conduct our simulations on real quantum computers (D-Wave Advantage) using a real dataset of vessel locations in the US. Numerical results show that our algorithm outperforms commercialized solutions of D-Wave by allowing existing quantum annealers to solve $17.5$ times larger BP instances while maintaining high solution quality. Although quantum computing cannot theoretically overcome the hardness of BP problems, this work contributes early efforts to applying quantum computing in satellite optimization problems, especially applications formulated as clique cover/graph coloring problems.

\end{abstract}
\section{Introduction}

Satellite communication systems are undergoing a business and technological evolution linked to the growing interest in Low-Earth Orbit (LEO) systems. With their ability to provide response to the soaring high speed demand in remote locations, their shorter round trip delay and their reduced manufacturing and launching cost, LEO satellite communication systems are currently being deployed in large scale \cite{9755278}.

The design and management of satellite communication systems typically results in complex optimization problems involving different degrees of freedom \cite{Flor_2022}. Focusing on the LEO satellite scenario, one of the key design challenges is the so-called ``Beam Placement'' (BP) problem, which seeks an optimal user-to-beam allocation. 
While well-established Geostationary (GEO) satellites ensure coverage over specific region via a fixed and regular grid of equally-spaced beams, in the LEO case a dynamic beam placement needs to be calculated in order to point the satellite beams towards the users on Earth. Operators may prefer to reduce the number of beams generated on-board so that the beamforming network is kept as simple as possible. The works in \cite{pachler2021static,Alinque_2020} represent the latest approaches to the beam placement problem. Both works faced the trade-off between minimizing the number of beams and maximizing the provided gain to users. 
However, what is relevant here is that both \cite{pachler2021static} and \cite{Alinque_2020} formulate the beam placement problem as a minimum clique cover problem, which is an NP-hard problem and therefore intractable for a classical computer to solve. Given the emergence of commercial quantum computers, we investigate the possibility of using quantum computing for this problem.

Quantum computers are expected to execute challenging computational tasks significantly faster than any classical computers. Instead of
processing information with classical binary bits, \textit{qubits} are used in quantum computers which are able to explore combinations of quantum states simultaneously  by leveraging superposition of quantum states\cite{KishorRevModPhys2022}. This reveals the parallel processing capability of quantum computing. Early theoretical results showed promising achievements. For examples, Shor's algorithm \cite{shor1994algorithms} can break the RSA encryption and  Grover's algorithm \cite{grover1996fast} can quadratically speed up unstructured search problems. There are two main quantum processor technologies, the Quantum Gate processors, such as IBM, IonQ, and Rigetti \cite{Elijah2022}, and Quantum Annealers, D-Wave \cite{Dwave}. The former is to make general purpose quantum computers while the later is designed specifically to solve optimization problems. To solve combinatorial problems, Quadratic Unconstrained Binary Optimization (QUBO) or  Ising formulations  are used as the formulation language \cite{AndrewIsingModel}. Then, Hamilitonians are prepared as inputs of quantum computers. Currently, only Quantum Annealers have large enough qubits for solving various real-life applications such as  car production \cite{YarkoniCQE2021}, 
 air traffic management \cite{StollenwerkTITS2020}, computer vision \cite{Dzung22hqc}  and wireless communications \cite{MinsungSigCom2019}.

However, current quantum annealer technologies has a major limitation. 
Even with $5000$+ physical qubits in latest D-Wave Advantages processors, the current quantum hardwares limit how large an optimization problem could be solved. For example, existing D-Wave Advantage can only solve Quadratic Assignment Problems with less than $200$ variables \cite{Kuramata2021}. 
A few approaches have been developed to tackle this issue.
Roof Duality, which is the most well-known method, implemented in D-Wave Ocean SDK \cite{RotherCVPR2007,Dwave}, is a linear relaxation based method that partially assigns values for a subset of binary variables, hence reducing the number of unknown variables. D-Wave also implemented Qbsolve, which decomposes large QUBO instances into smaller ones. However, this tabu-search-based function was no longer supported after March 2022 \cite{Dwave}. 
A very recently proposed approach, FastHare, finds non-separable groups in which variables have identical values in the optimal solution
and merges the variables of a group into one logical qubit \cite{ThaiCQE2022}. 
As observed through extensive experiments carried out in \cite{ThaiCQE2022}, the aforementioned approaches seemed to work well on a subset of QUBO instances while did not fit others.
As such, the relevance of those approaches to satellite communication problems, including beam placements, requires a rigorous investigation.
In addition, in the aforementioned works, their original problems are QUBOs, unconstrained formulations. However, BP problems are constrained. As a result, the feasibility of the original problems has to be further investigated. 

In this paper, we first introduce the QUBO formulation of BP problems which can be used to fit in quantum annealers. To tackle quantum annealers' variable limit, rather than preparing Hamiltonians from QUBO formulations of an original Beam Placement optimization instance, we propose an efficient Hamiltonian reduction method to construct a small-size ``equivalent'' Hamiltonian which can fit into quantum annealers. This method includes two parts: variable presolve and Hamiltonian formulation. The former 
is a two-step presolve method based on graph operations and linear relaxations, which 
could presolve most variables of the Beam Placement instance. Then, using the remaining variables, a reduced Hamiltonian is constructed as an input of quantum annealers. 
  We perform the first medium-scale benchmarks for the proposed process on various aspects based on US vessel location data assuming that each vessel is a satellite communication subscriber. We first highlight that our method can significantly reduce the number of logical qubits needed to construct a Hamiltonian. It can presolve more than $99\%$ variables in most of cases while D-Wave Roof Duality cannot presolve any variables. Moreover, our solutions allows D-Wave Advantages to solve problems whose number of users is $17.5$ times higher while maintaining the feasibility of the final solutions at $91.5\%$. Last but not least, numerical results returned by D-Wave Advantage based on our reduced Hamiltonians show better beam placement performance as compared with classical counterpart, indicating the potentials of applying quantum annealing for satellite problems, especially 
  when formulated as a clique cover problem. Since a clique cover of a graph is a vertex coloring of its complement and vice versa \cite{Karp72}, this approach also demonstrates the potential of solving applications of the clique cover and graph coloring problems in satellite networks. 


\section{Quantum Annealing and QUBO/Ising Formulation}

Quantum annealers including D-Wave's are quantum computers
that solve optimization problems through energy minimization of a physical system of $S$ qubits\cite{Dwave,AndrewIsingModel}. The quantum system's energy profile is defined by its Hamiltonian. In Quantum Annealing (QA) \cite{AndrewIsingModel}, the quantum system is first initialized at the ground state of the initial Hamiltonian $H_0$ and then slowly evolve the system
Hamiltonian to the targeted Hamiltonian $H_P$ in a given large duration $T$ according to 
\begin{align}
    \mathcal{H}(t) = \left (1 - \frac{t}{T} \right)H_0 + \frac{t}{T}H_P,
\end{align}
where,
\begin{align}
 H_0 = \sum_{i=1}^{S} \sigma_i^x,     H_P = \sum_{i=1}^S f_i \sigma_i^z + \sum_{i,j=1}^S G_{ij} \sigma_i^z \sigma_j^z.
\end{align}
where $f_i$ and $G_{ij}$ are system parameters called bias and coupler strength, respectively; $\sigma_i^z$ and $\sigma_i^x$ are Z and X Pauli operators on the $i$th qubit \cite{AndrewIsingModel}. By measuring the ground state of the system at time $T$, it is equivalent to find the optimal solution of an Ising Model such that
\begin{align}
    \min_{s_i \in \{-1,1\}^S} H_P(\mathbf{s})=\sum_{i=1}^{S} f_i s_i + \sum_{i,j=1}^S G_{ij} s_i s_j .
\end{align}

The Ising model can be alternatively represented as Quadratic Unconstrained Binary Optimization (QUBO) such that
\begin{align}
    \min_{\mathbf{x}\in \{0,1\}^S} ~~~& \sum_{i=1}^{S} Q_{ii}x_i + \sum_{i,j=1}^S Q_{ij} x_i x_j
\end{align}
where $\mathbf{x} = [x_1, \ldots, x_S]$ are binary variables and $\mathbf{Q} \in \mathbb{R}^{S\times S} $ is an upper triangular matrix. QUBO can be easily converted back and forth to an Ising
Hamiltonian by mapping $s_i = 2x_i -1$, $f_i = \frac{1}{2} Q_{ii} + \frac{1}{4}\sum_{j=1}^{S} Q_{ij} + \frac{1}{4}\sum_{j=1}^{S} Q_{ji}$ and $G_{ij} = \frac{1}{4} Q_{ij}$.




\section{System Model and Problem Formulation}
\subsection{System Model}
\begin{figure}[t]
		\subfigure[]{\includegraphics[width=0.5\linewidth]{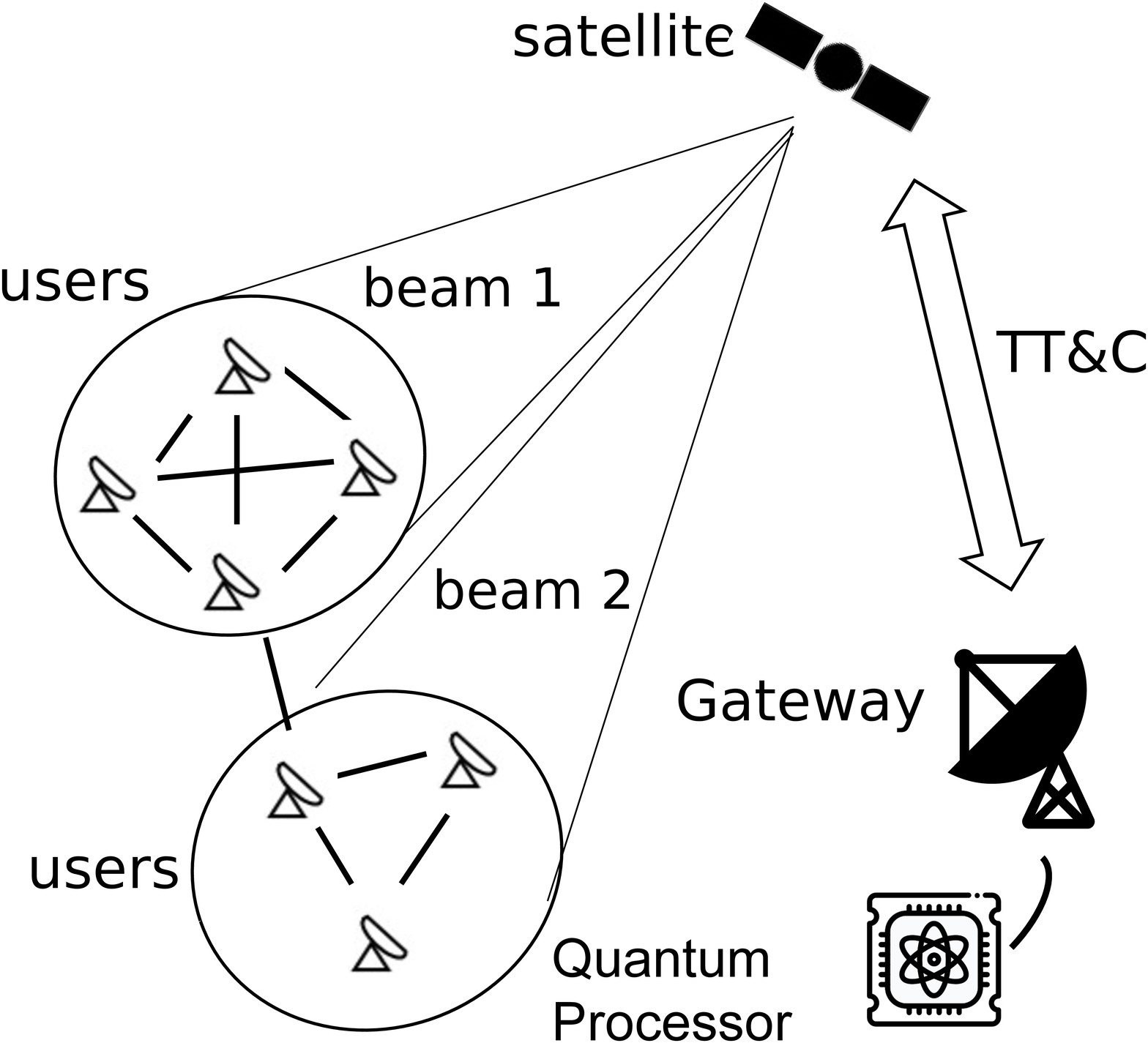}} 
		\hspace{0.2cm}
		\subfigure[]{\includegraphics[width=0.35\linewidth]{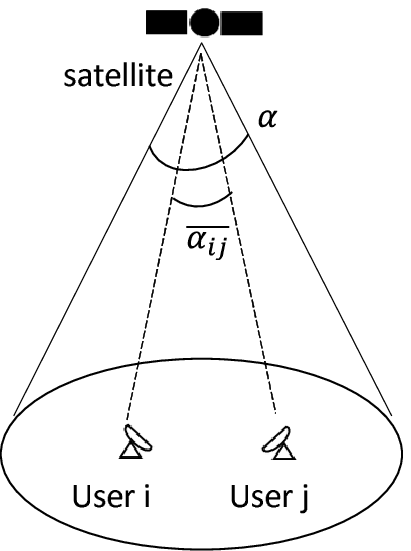}}		
		\caption{ (a) The System Model and (b) Example of one beam pattern
with two users in the beam \label{fig::system_model} }
	\end{figure}
Illustrated in Fig. \ref{fig::system_model}a, let us consider a system where there are $N$ users and a satellite equipped with an antenna array which is able to provide maximum $B$ beams, $B \leq N$, to serve its users. Let $\mathcal{N} = \{1, \ldots,N \}$ and $\mathcal{B}= \{1,\ldots,B \}$ denote the set of users indices and the set of beam indices, respectively. The satellite can access a gateway equipped with a quantum processor via  its TT \&C (telemetry, tracking, and control) link to conduct optimization operations. To model the BP problem, we make use of the User Assignment Matrix (UAM), which is defined as $\mathbf{A} = \{ (a_{ib}) \in \{0,1 \} \} $, where $a_{ib} = 1$, if user $i$ is assigned to beam $b$. Otherwise, $a_{ib} = 0$. 	

To exploit users' traffic diversity and maximize multicast benefits, the satellite aims at minimizing the number of active beams while ensuring the requested service by its users \cite{pachler2021static,Alinque_2020}. To do that, nearby users can be served by the same beam. Leveraging the results of previous works \cite{pachler2021static,bui2022joint}, to be served in the same beam, with the specific location of the satellite, any pairs of a given set of users must be separated by an angle $\Bar{\alpha} \leq \frac{\alpha}{2}$, where $\alpha$ is the cone angle of a beam. This is illustrated in Fig. \ref{fig::system_model}b. Consequently, depending on the positions of the satellite on the sky in different time slots, we have different undirected proximity graphs $G = (\N,E)$, where edge between user $i$ and user $j$ exists if $\Bar{\alpha}_{ij} \leq \frac{\alpha}{2}$. Let $z_b$ denote the indicator of whether beam $b$ is active or not. The BP problem of finding the minimum number of active beams is formulated as followed,
\begin{align}
    \textrm{P1:}\min_{\{z_b\},\{a_{ib}\}} ~~& \sum_{b \in \mathcal{B}} z_b,\\
    \textrm{s.t.}~~ & \text{C1:} ~ \sum_{b \in \mathcal{B}} a_{ib} = 1, \forall i \in \mathcal{N}, \label{constraint:C1}\\
    & \text{C2:} ~ a_{ib} \leq z_b, \forall i \in \mathcal{N}, \forall b \in \mathcal{B} \label{constraint:C2},\\
    & \text{C3:} ~ a_{ib} + a_{jb} \leq 1, \forall b, ~\textrm{if}~ (i,j) \notin E,\label{constraint:C3}\\
    & \text{C4:} ~ \sum_{i \in \mathcal{N}} a_{ib} \leq W, \forall b \in \mathcal{B} \label{constraint:C4},\\
    & \text{C5:} ~  z_b, a_{ib} \in \{0,1\}. \label{constraint:C5}
\end{align}

Here, $C1$ indicates that each user must be assigned in only one beam. $C2$ indicates that no user can be served by an inactive beam. $C3$ indicates that the angle separating two users which respects to the satellite must be smaller than $\alpha/2$ to be in the same beam. $C4$ indicates the maximum number of users that can be served by one beam.

\subsection{QUBO formulation}
Based on prior knowledge on how to bring constraints into objective \cite{AndrewIsingModel}, P1 is transformed into QUBO as follows. Due to bin-packing constraints C4, binary slack variables $s_{bw}, w \in \mathcal{W} =  \{1,\ldots,W \}$ are introduced to make $C4$ become equality. $P1$ is then transformed into $\Bar{P1}$:
\begin{align}
    \Bar{\textrm{P1}}:\min_{\substack{\{z_b\},\{a_{ib}\},\\ \{s_{b,w}\}}} ~~& \sum_{b \in \mathcal{B}} z_b,\\
\nonumber   \textrm{s.t}~~ & (\ref{constraint:C1}) - (\ref{constraint:C3}), (\ref{constraint:C5}) \\
    &  \sum_{i \in \mathcal{N}} a_{ib} + \sum_{w \in \mathcal{W}} ws_{bw} = W, \forall b \in \mathcal{B},
\end{align}
Let $\mathbf{x} = [\mathbf{a}_1^T,\ldots,\mathbf{a}_B^T,z_1,\ldots,z_B,\mathbf{s}_{1},\ldots,\mathbf{s}_{B}]^T$ denote the column vector representing variables of $\Bar{P1}$.
 First, denote $\mathbf{Q}_{C1},\mathbf{Q}_{C2},\mathbf{Q}_{C3}$ and $\mathbf{Q}_{C4}$ the matrices corresponding to C1, C2, C3 and C4. Denote $\mathbf{Q}_o$ the matrix corresponding to the objective function. We have 
\begin{align}
    \textrm{QUBO-1}:\min_{\mathbf{x}} ~~& \mathbf{x}^T \mathbf{Q} \mathbf{x},
\end{align}
where $\mathbf{Q} = \mathbf{Q}_o + \lambda \big(\mathbf{Q}_{C1}+\mathbf{Q}_{C2}+\mathbf{Q}_{C3} + \mathbf{Q}_{C4}\big)$ and
\begin{align}
  \mathbf{Q}_o & = \textrm{diag}\Big(\Big[\mathbf{0}_{1\times(NB)},\mathbf{1}_{1\times B},\mathbf{0}_{1\times(WB)} \Big]\Big), \\
     \mathbf{Q}_{C1} & = \mathbf{A}_1^T\mathbf{A}_1 - 2\textrm{diag}\Big( \mathbf{d}_1^T\mathbf{A}_1\Big), \\
  \nonumber   \mathbf{Q}_{C2} &= \textrm{diag}\Big(\Big[\mathbf{1}_{1\times(NB)},\mathbf{0}_{1\times(B+WB)} \Big]\Big)
    \\
    & ~~~+ \begin{bmatrix}
     \mathbf{0}_{NB \times NB},\hspace{0.7cm} \mathbf{A}_2,\hspace{1.4cm}\mathbf{0}_{NB\times WB}\\
     \mathbf{0}_{(B+BW) \times NB},\mathbf{0}_{(B+BW)\times B},\mathbf{0}_{(B+BW)\times WB}
     \end{bmatrix}, \\
    \mathbf{Q}_{C3} &= \textrm{diag}\Big(\Big[\underbrace{ \overline{\mathbf{F}}, \ldots ,\overline{\mathbf{F}}}_{B~\textrm{times}}, \mathbf{O}_{(B+WB)\times(B+WB)}\Big]\Big),\\
     \mathbf{Q}_{C4} & = \! \mathbf{A}_3^T\mathbf{A}_3 - 2\textrm{diag}\Big( \mathbf{d}_2^T\mathbf{A}_3\Big),
\end{align}
where
\begin{align}
    \mathbf{d}_1 & = \mathbf{1}_{N \times 1}, \\
    \mathbf{A}_1 & =[\underbrace{ \mathbf{I}_{N \times N}, \ldots ,\mathbf{I}_{N \times N}}_{B~\textrm{times}}, \mathbf{0}_{N \times B}, \mathbf{0}_{N \times W}],  \\
    \mathbf{d}_2 & = W\mathbf{1}_{N \times 1}, \\
  \mathbf{A}_2 &= \begin{bmatrix}
  -\mathbf{1}_{N \times 1},~\mathbf{0}_{N \times 1}~,\ldots,~~\mathbf{0}_{N \times 1}\\
  ~\mathbf{0}_{N \times 1},-\mathbf{1}_{N \times 1}~,\ldots,~\mathbf{0}_{N \times 1}\\
   \vdots \\
 \underbrace{ ~\mathbf{0}_{N \times 1},~\mathbf{0}_{N \times 1},~\ldots, -\mathbf{1}_{N \times 1}}_{B~\textrm{columns}}
  \end{bmatrix},\\
     \mathbf{A}_3 &=   \begin{bmatrix}
    \mathbf{1}_{1 \times B}, \mathbf{0}_{1 \times B}, \ldots, \mathbf{0}_{1 \times B},\mathbf{0}_{1 \times B},1,\ldots,W\\
   \mathbf{0}_{1 \times B}, \mathbf{1}_{1 \times B}, \ldots, \mathbf{0}_{1 \times B}, \mathbf{0}_{1 \times B},1,\ldots,W\\
   \vdots \\
   \underbrace{\mathbf{0}_{1 \times B}, \mathbf{0}_{1 \times B}, \ldots, \mathbf{1}_{1 \times B} }_{NB~\textrm{columns}}, \mathbf{0}_{1 \times B}, 1,\ldots,W
  \end{bmatrix}.
\end{align}
Here, $\mathbf{I}$ is an identity matrix with ones on the main diagonal and zeros elsewhere, and $\overline{\mathbf{F}}$ is the $N$-by-$N$ adjacency matrix of the complement of $G$. Moreover, $\textrm{diag}(\cdot)$ creates a (block) diagonal matrix whose main diagonal is its argument. Based on 
\cite{AndrewIsingModel}, we set the value of $\lambda = B+1$.
Then, to obtain upper-triangular matrix $\mathbf{Q}$ accepted by D-Wave solvers, $Q_{ij} = Q_{ij} + Q_{ji}, \forall j>i$ and $Q_{ij} = 0, \forall j<i$.

\subsection{Curse of Dimensionality}
In contrast to classical computers, which requires $NB+B$ variables to solve $P1$, to construct QUBO-1 for quantum annealer, $P1$ requires $NB +B + WB$ logical qubits, which may not fit to recent quantum computer (D-Wave), in which $5000+$ physical qubits are required to represent less than $200$
logical qubits \cite{DwaveUplimit}. Thus, even when we consider a small instance with $12$ users and $12$ beams where each beam can cover $5$ users, the total number of logical qubits required is $216$, which already exceeds D-Wave's capacity. We can only use D-Wave Advantage when $N = 10$, $B=10$ and $W=5$.


\section{Hamiltonian Reduction}
In order to reduce the size of Hamiltonians, in this section, we propose a solution consisting of variable presolve and reduced Hamiltonian formulation. For variable presolve, we first find a subset of users who cannot be served by the same beam by identifying a maximum independent set in the proximity graph~$G$. We assign each of such users to a different beam. 
Next, we obtain $P2$ as a reduced version of $P1$ with remaining users and solve a linear relaxation of $P2$, which successfully assigns many users $i$ to their respective beams $b$ (when $a_{ib}=1$). Finally, we construct a reduced Hamiltonian and use QA to find beam allocations for the unassigned users. 

\subsection{Two-step Variable Presolve}
\subsubsection{Find Maximum Independent Set}
Let $\N_{\textrm{ind}}$ be the set of users corresponding to a large independent set of $G$.
Note that finding a maximum independent set of a graph is an NP-hard problem. We can use any existing method in the literature to find an approximation of a maximum independent set instead, e.g., the one by Boppana and Halldórsson~\cite{Boppana1992}. However, in this work, as a proof-of-concept, we employ a simple and fast greedy algorithm to find $\N_{\textrm{ind}}$ (see Algorithm~\ref{algorithm:find_max_ind_set}).

		\begin{algorithm}    [t]                
			\caption{Finding an independent set \label{algorithm:find_max_ind_set}}          
			{
				\begin{algorithmic}[1]                    
					\Require $G = (\N,E)$ - the proximity graph 
					\Ensure $\N_{\textrm{ind}}$ - an independent set of $G$
					\State $\N_{\textrm{ind}} = \emptyset$,
					\While {$\N \neq \emptyset$}
					\State $n^*= \arg \min_{n \in \N}\textrm{Node Degree}(n)$, 
					\State  $\N_{\textrm{ind}} = \N_{\textrm{ind}} \cup  \{n^*\}$, 
					\State $\N = \N \setminus\big( \{n^*\} \cup \textrm{Neighbor Set}(n^*)\big)$,
					\EndWhile
				\end{algorithmic}}
			\end{algorithm}

Once pre-assigning users in $\N_{\textrm{ind}}$ to $|\N_{\textrm{ind}}|$ different beams, with the set of remaining users $\N \setminus \N_{\textrm{ind}}$, 
 the problem $P1$ is now reduced to the following problem.
\begin{align}
    \textrm{P2:}\min_{\{z_b\},\{a_{ib}\}} ~~& \sum_{b\in \mathcal{B}} z_b,\\
 \nonumber   \textrm{s.t}~~ & (\ref{constraint:C1}) - (\ref{constraint:C4}), \\
    &  z_b, a_{ib} \in \{0,1\}, \forall i,j \in \N \setminus \N_{\textrm{ind}}.
\end{align}


\begin{lemma} \label{lemma:max_ind_set}
Let $\mathrm{OPT}(P1)$ and $\mathrm{OPT}(P2)$ denote the optimal value of $P1$ and $P2$. By pre-assigning users in $\N_{\textrm{ind}}$ to $|\N_{\textrm{ind}}|$ different beams, the optimal solution in $P2$ is still optimal in $P1$, i.e., $\mathrm{OPT}(P1)= \mathrm{OPT}(P2)$.
\end{lemma}
\begin{proof}
Assigning one separate beam to each user corresponding to vertices in an independent set doesn't affect the optimality of the solution because these users must be allocated to different beams in any feasible solution.
\end{proof}

\subsubsection{Linear Relaxation}
We consider a linear relaxation $P2'$ of $P2$, which can be efficiently solved by simplex method \cite{nelder1965simplex}, given below,
\begin{align}
    P2':\min_{\{z_b\},\{a_{ib}\}} ~~& \sum_{b\in\mathcal{B}} z_b,\\
 \nonumber   \textrm{s.t}~~ & (\ref{constraint:C1}) - (\ref{constraint:C4}), \\
    & z_b, a_{ib} \in [0,1], \forall i,j \in \N \setminus \N_{\textrm{ind}}.
\end{align}

\begin{theorem}
Let $\mathrm{OPT}(P2')$ denote the optimal value of $P2'$. Then $\mathrm{OPT}(P2') \leq \mathrm{OPT}(P2) = \mathrm{OPT}(P1)$.
\end{theorem}
\begin{proof}
Since $P2'$ is the linear relaxation formulation of $P2$, $\mathrm{OPT}(P2') \leq \mathrm{OPT}(P2)$. From Lemma \ref{lemma:max_ind_set}, $\mathrm{OPT}(P1)= \mathrm{OPT}(P2)$. Thus, $\mathrm{OPT}(P2') \leq \mathrm{OPT}(P2) = \mathrm{OPT}(P1)$.
\end{proof}
\begin{remark} 
Let $P1'$ and $\mathrm{OPT}(P1')$ denote the linear relaxation formulation of $P1$ and its optimal value.  Given an algorithm $\mathrm{ALG}$ with its returned value $v^{\mathrm{ALG}}$, 
\begin{align}
    \frac{v^{\mathrm{ALG}}}{\mathrm{OPT}(P1)} \leq \frac{v^{\mathrm{ALG}}}{\mathrm{OPT}(P2')} \leq \frac{v^{\mathrm{ALG}}}{\mathrm{OPT}(P1')}. 
\end{align}
Thus, by solving $P2'$, which has fewer unknowns than $P1'$, we obtain a tighter bound for the approximation ratio $\frac{v^{\mathrm{ALG}}}{\mathrm{OPT}(P1)}$.
\end{remark}
\begin{proof}
Since the domain of $P2$ is a subset of the domain of $P1$, the domain of $P2'$ is a subset of $P1'$. As a result, $\mathrm{OPT}(P1') \leq \mathrm{OPT}(P2') \leq \mathrm{OPT}(P1)$.
\end{proof}

\subsection{Reduced Hamiltonian Formulation}
After solving $P2'$, we obtain an optimal solution $\Bar{\mathbf{x}}$. 
Users $i$ that have $a_{ib}=1$ for some $b$ are allocated to that beam $b$.
Let $\Bar{\mathcal{N}}$ be the set of the unassigned users, i.e., users $i$ with $a_{ib}\neq 1$ for every $b\in\mathcal{B}$. Let $\tilde{\mathcal{B}}$ be the set of active beams ($z_b=1$). Due to the symmetry of beam variables in the problems, we could re-index the active beams from $1$ to $|\tilde{\mathcal{B}}|$. We then construct the Hamiltonian that is 
a reduced version of the QUBO formulated from $P1$ as follows.

The maximum number of more beams that could be active (other than the first $|\tilde{\mathcal{B}}|$ beams) is $\Bar{B} = \min \{B - \tilde{\mathcal{B}}, \Bar{\mathcal{N}}\}$. Thus, we only include $z_b, b = |\tilde{\mathcal{B}}|+1,\ldots,\Bar{B}$ in the reduced Hamiltonian. Similarly, for beam association variables, we only include $a_{ib},i \in \Bar{\mathcal{N}} , b = \tilde{\mathcal{B}}+1,\ldots,\Bar{B}$. For slack variables corresponding to active beam $b \in \tilde{\mathcal{B}}$ whose users are smaller than $W$, we only keep slack variables corresponding to the remaining capacity of that beam. Let $\overline{W}_b$ denote the remaining capacity of beam $b$. For example, if the remaining capacity of a given beam $b'$ is $3$, i.e, $\overline{W}_{b'} = 3$,  we only keep slack variables $s_{b'1},s_{b'2}$ and $s_{b'3}$. To further reduce the number of variables in the reduced Hamiltonian, we eliminate all slack variables corresponding to $b=|\tilde{\mathcal{B}}|+1,\ldots,\Bar{B}$ by assuming that the largest clique of the remaining users has size at most $W$.
Note that in some cases where this assumption is violated, the returned solution is infeasible, which reduces the success probability of QA's solutions (see Fig.~\ref{fig::success}). 
The matrix $Q'$ corresponding to the reduced Hamiltonian $H_P'$ is then created by selecting the rows and columns of $Q$ corresponding to the aforementioned variables.  

The reduced Hamiltonian $H_P'$ is sent to D-Wave via its API. Once API returns results, we choose the solution with the lowest energy level. 
We then merge the newly solved variables together with those determined earlier when solving $P2'$ to obtain a complete solution to $P1$ (see Algorithm~\ref{algorithm:annealing_with_reduced_Q}). 
		\begin{algorithm}    [t]                
			\caption{Quantum Annealing with Reduced Hamiltonian}          
			\label{algorithm:annealing_with_reduced_Q}                           
			{
				\begin{algorithmic}[1]                    
					\Require $\Bar{\mathbf{x}}$, $Q$, where $\Bar{\mathbf{x}}$ is an optimal solution of $P2'$
					\Ensure $\mathbf{x}^{\textrm{QA}}$ - a solution to $P1$
					\State Construct $\Bar{\mathcal{N}}$ - the set of unassigned users
					\State Construct $Q'$ by selecting rows and columns of $Q$ corresponding to $z_b, b = |\tilde{\mathcal{B}}|+1,\ldots,\Bar{B}$, $a_{ib},i \in \Bar{\mathcal{N}} , b = \tilde{\mathcal{B}}+1,\ldots,\Bar{B}$ and $s_{bw}, w = 1, \ldots, \overline{W}_{b}, b \in \tilde{\mathcal{B}}$
					\State Once D-Wave has converted $Q'$ to $H_P'$, send $H_P'$ to D-Wave Advantage processor
					\State Once received D-Wave results, return $\mathbf{x}^{\textrm{QA}}$, which is obtained by combining the results of D-Wave with  $\Bar{\mathbf{x}}$ 
				\end{algorithmic}}
			\end{algorithm}
\section{Numerical Results}

	In this section, we carry out  simulations to evaluate the
		performance of the proposed frameworks. To model the users, we make use of real vessel locations collected by the U.S. Coast Guard on the first of January 2022 \cite{Vesseldata}. We randomly sample locations from unique vessels in the region of Gulf of Mexico with the sample size from $25$ to $200$. We also set the altitude of the satellite at $1110$ km, which is the most popular LEO altitude of Starlink \cite{pachler2021static}. The latitude and longitude of satellite are $26.812309$ and $-85.386382$, respectively. The maximum number of beams is set at the number of users. The beam's capacity $W$ is set at $20$. For each value of number of users, we will sample $200$ realizations. 
					\begin{table}[t] \footnotesize			
						\caption{Default Parameter Setup.} \label{tab:para_setup_2}
						\begin{center}
							\begin{tabular}{|c|c| }
								\hline
								{\bf Parameter} & {}{\bf Value} \\  \hline
								$N$ & $25$--$200$  \\  
								Satellite's altitude & $1110$ km \\  
								Satellite's latitude and longitude & $26.812309$, $-85.386382$ \\  
								Number of realization & $200$ \\ 
								\hline	
							\end{tabular}
						\end{center}%
                        \vspace{-10pt}
					\end{table}%
	\begin{figure}[t]
		\centering
		{\includegraphics[width=0.9\linewidth]{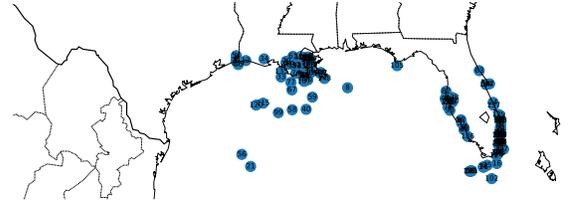}}
		\caption{Locations of $150$ random vessels at U.S south coast. \label{fig::map} }
	\end{figure}
 
\subsection{Reduction Ratio}
	\begin{figure}[t]
		\centering
		{\includegraphics[width=0.7\linewidth]{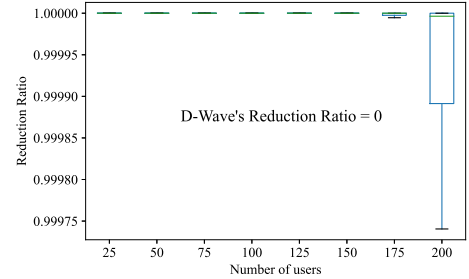}}
		\caption{Distribution of Reduction Ratio of the proposed scheme. In our experiments with 200 users, in $75\%$ of the trials, the reduction ratios are greater than $99\%$. On the contrary, D-Wave Roof Duality provides no reduction. \label{fig::reduction} }
  \vspace{-10pt}
	\end{figure}
In this part, we evaluate how much the size of the Hamiltonian is reduced. We compare the proposed scheme with the Roof Duality implemented in D-Wave Ocean SDK \cite{Dwave}. Here, the reduction ratio of the problem is defined as 
\begin{align}
    \textrm{Reduction ratio} = 1 - 
    \frac{S'}{S},
\end{align}
where $S'$ and $S$ are the number of logical qubits for the reduced Hamiltonian and the original one, respectively. In order words, $S \times S$ and $S' \times S'$ are the size of the original and the reduced Hamiltonian, respectively. We 
observe that in our experiment, D-Wave Roof Duality does not reduce any variables. 
By contrast, as can be seen in Fig. \ref{fig::reduction}, our variable presolve step obtains significant reduction ratios. We also observe that when the number of users increases, the reduction ratio of the reduction scheme decreases, which is perhaps due to the   increase of the problem's complexity.
However, the first quantile of reduction ratio is always larger than $99\%$ in the simulation settings.

\subsection{Solution's Feasibility}
	\begin{figure}[t]
		\centering
		{\includegraphics[width=0.8\linewidth]{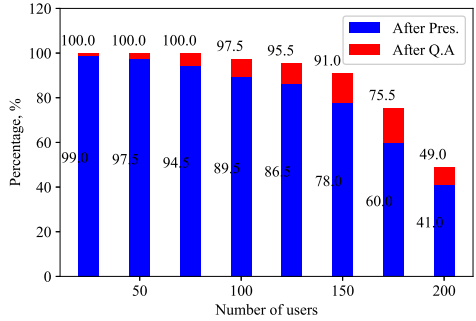}}
		\caption{Success probability to achieve feasible solutions of our scheme. Note that D-Wave BQM and CQM cannot solve those instances without variable reduction methods due to hardware limit.  \label{fig::success} }
  \vspace{-10pt}
	\end{figure}
In this part, we would like to evaluate the feasibility of solutions returned by our proposed scheme. Here, the \textit{success probability} is defined as the chance the final solution is feasible in the original problem $P1$. In Fig. \ref{fig::success}, cases in which users are completely assigned by solving $P2'$ without the need of using QA are colored in blue and labeled as ``Pres.'' while cases where QA is used to decide unsigned users are colored in red and labeled ``Q.A''. We observe that our variable presolve scheme can solve most of the network instances with number of users smaller than $175$. Moreover, by using QA to solve the undecided variables, we can solve up to $15.5\%$ of the network instances that the reduction scheme does not completely solve. We note that D-Wave Binary Quadratic Model (BQM) and Constrained Quadratic Model (CQM) cannot solve those instances due to their large number of variables. We observe that as the problem's complexity increases when the number of users increases, the chance that the derived solution in our scheme is feasible decreases, especially when there are more than $175$ users. Still, a $175$-user instance is already $17.5\times$ greater than the toy example that D-Wave can handle (only $10$ users) discussed in Section III.C. Further improvements of our scheme are left to the future works.
\subsection{Algorithm Performance}
		\begin{figure}[t]
		\centering
\includegraphics[width=0.9\linewidth]{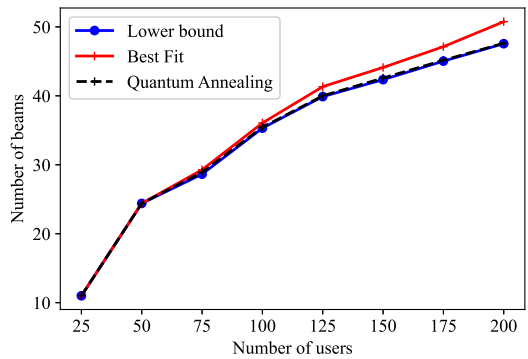}
		\caption{The performance comparison of our proposed algorithm and Best Fit~\cite{silberschatz2018operating}. Our algorithm achieves near-optimal solutions and outperforms Best Fit (solved on a classical computer). \label{fig::performance} }
  \vspace{-10pt}
	\end{figure}
In this part, we would like to evaluate the performance of Quantum Annealing in comparison with a well-known classical algorithm, Best Fit \cite{silberschatz2018operating}. 
To solve BP, Best Fit simply
assigns a user to the most utilized beam that the user can be assigned to without violating any constraints. If the user can not be assigned to any existing beams, the satellite activates a new beam. Here, for fair evaluation, we only consider cases that are completely solved by Quantum Annealing. We remove cases where the reduction scheme can reach optimality. We also use the objective value of $P2'$ as the lower bound of the optimal beam placement. In Fig. \ref{fig::performance}, we can see that for the cases where Quantum Annealing's solutions are feasible, their objective values can reach near-optimal quality. Moreover, we also observe that the gap between Quantum Annealing and Best Fit increases as the number of users increases.

\section{Conclusion}
In this paper, we use a quantum computer to solve the Beam Placement problem and evaluate the performance on a real-world dataset of vessel locations in the US.
We first construct a QUBO formulation for the problem as an input of quantum annealers. Due to quantum annealers' variable limit,
we propose a Hamiltonian reduction method to reduce the number of logical qubits required. 
First, our reduction method significantly reduces the number of logical qubits required to solve the equivalent QUBO of the original problem with the first quantile of reduction ratio more than $99\%$ while D-Wave Roof Duality offers no reduction. Second, the experiments show that there is a high chance to reach feasible solutions using our reduction method and Quantum Annealing while D-Wave BQM and CQM cannot solve the instances with more than $25$ users. Finally, we note that feasible solutions achieved by our method outperform Best Fit, a well-known classical method. We hope that our work helps trigger further developments on applying quantum computing for satellite communications problems, especially those with clique covering/graph coloring formulations.

\section{Acknowledgement}
This work has been supported by the Luxembourg National Research Fund (FNR) under the project FlexSAT (C19/IS/13696663) and the project MegaLEO (C20/IS/14767486).
		\bibliographystyle{IEEEtran}
		\bibliography{IEEEabrv,references}
\end{document}